\documentclass[conference]{IEEEtran}
\IEEEoverridecommandlockouts
\usepackage{cite}
\usepackage{amsmath,amssymb,amsfonts}
\usepackage{mathtools}

\usepackage{cuted}
\usepackage{graphicx}
\usepackage{textcomp}
\usepackage{xcolor}
\usepackage{cuted}
\usepackage{amsmath, amssymb, amsthm,algorithm}
\usepackage[latin1]{inputenc}

\usepackage{latexsym}
\usepackage{graphics,epsfig}
\usepackage{amsfonts}
\usepackage{booktabs}
\usepackage{color}
\usepackage{multirow}
\usepackage[justification=centering]{caption}
\usepackage{dsfont}
\usepackage{graphicx}
\usepackage{adjustbox}
\usepackage{kantlipsum}
\usepackage[caption=false,font=footnotesize]{subfig}
\usepackage{cases}
\usepackage{url}
\usepackage{tabu}
\usepackage{makecell}
\usepackage{soul}

\usepackage{array}
\usepackage{algpseudocode}
\usepackage{mathtools,lipsum,cuted}

\usepackage[justification=centering]{caption}

\usepackage[caption=false,font=footnotesize]{subfig}
\usepackage[T1]{fontenc}
\usepackage{mwe}
\usepackage{subfig}

\newlength{\tempheight}
\newlength{\tempwidth}

\newcommand{\rowname}[1]
{\rotatebox{90}{\makebox[\tempheight][c]{#1}}}

\newcommand{\columnname}[1]
{\makebox[\tempwidth][c]{#1}}

\makeatletter
\def\BState{\State\hskip-\ALG@thistlm}
\makeatother

\newcounter{example}[section]

\usepackage[english]{babel}
\usepackage{verbatim}

\usepackage[english]{babel}

\theoremstyle{plain}

 \newtheorem{prop}{Proposition}

\theoremstyle{remark}

\def\BibTeX{{\rm B\kern-.05em{\sc i\kern-.025em b}\kern-.08em
    T\kern-.1667em\lower.7ex\hbox{E}\kern-.125emX}}
\begin{document}

\voffset=0.05in
\textheight=9.28in

\title{Outage Probability Analysis of  Wireless Paths with Faulty Reconfigurable Intelligent Surfaces}

\author{\IEEEauthorblockN{Mounir Bensalem and Admela Jukan}
\IEEEauthorblockA{Technische Universit\"at Braunschweig, Germany;
\{mounir.bensalem,  a.jukan\}@tu-bs.de}

}

\maketitle

\begin{abstract}
 We consider a next generation wireless network incorporating a base station a set of typically low-cost and faulty Reconfigurable Intelligent Surfaces (RISs). The base station needs to select the path including the RIS to provide the maximum signal-to-noise ratio (SNR) to the user. We study the effect of the number of elements, distance and RIS hardware failure  on the path outage probability, and based on the known signal propagation model at high frequencies, derive the closed-form expression for the said probability of outage. Numerical results show the path outage likelihood as function of the probability of hardware failure of RIS elements, the number of elements, and the distance between mobile users and the RIS.
\end{abstract}

\section{Introduction}
Reconfigurable intelligent surfaces (RISs) are becoming an integral system component in millimeter-wave (mmW) and Terahertz (THz) communications, aimed at solving the problem of considerable volatility of transmission due to physical obstacles by reflecting the wireless links around the obstacles\cite{boulogeorgos2021pathloss, wu2021intelligent}.  In addition, RIS partitioning into sub-surfaces, whereby each sub-surface consists of a number of adjacent reflecting elements, allows the creation of multiple paths for multi-user and reliability scenarios \cite{mao2022element, cai2023ris}. On the other hand, the susceptibility of individual RIS elements, referred to as unit cells or meta-atoms, to failures is already a well-known concern  \cite{li2021array}. Faulty RIS elements can lead to degradation in the antenna radiation pattern, and in more severe cases, impact the entire functionality of RIS meta-surface elements. To ensure reliable RIS-aided networks in practical scenarios, it is critical to design link and path selection schemes that account for outage due to RIS meta-surface failures. 

In this paper, we provide a novel analysis of the outage probability of wireless paths due to a) RIS hardware failure and b) connection degradation due to obstacles, both stationary and mobile. We analyse and obtain a closed-form of the instantaneous channel state information (CSI) model of an end-to-end communication system including base station as transmitter, RIS (as an intermediate node) and user as receiver.  In our approach, the base station does not need to consider an instant CSI, and the "best" RIS can be selected based on outdated CSI. We formulate the SNR model using the instantaneous CSI and analyse the partial best path selection. 
We derive a closed-form expression of the outage probability for an arbitrary path in terms of the modified Bessel function of the second kind.  We also study the effect of RIS hardware failure. Thus, the closed-form expression is dependent on the failure probability of RIS elements, the existence of physical obstacle, devices characteristics, and the length of the individual links.
We show that the analysis can be used to quantify the impact of the  blockages and  penetration losses as well as the RIS hardware failure on the outage probability for arbitrary paths. We show that at a signal-to-noise-ratio (SNR) approaching zero,  the presence of obstacles and the failure probability of RIS elements majorly affect the outage probability. A more reliable connection can be obtained at a large enough RIS device vs. a communication distance, in consideration of the RIS with strongly correlated elements and their likelihood of failure. This is due to the outage probability having a comparably lower values given a certain minimum number of RIS meta-surface elements. 
\par The rest of the paper is organized as follows. Section II presents the related work. Section III provides the system model and Section IV the outage probability analysis. Section V discusses numerical results. Section VI concludes the paper.

\section{Related Work}

Several studies focused on the outage probability in RIS-aided communications, including \cite{yang2020outage, wang2021outage, lu2021outage, jung2019reliability}. 
In \cite{yang2020outage}, the authors has analysed the outage probability of a RIS-aided system, where the RIS with highest SNR value is selected. Their results showed that the number of RISs and the number of reflecting elements are crutial for  the capacity scaling law of multiple RIS-aided networks. In our work, a partial path selection approach was adopted, meaning that the paths with the best SNR value between the RIS and the user are selected to support the communication.
In \cite{jung2019reliability},  a large intelligent surfaces (LISs), which can be placed in structures like walls,   has been studied as a solution for reliable communication. The outage probability of LIS was characterized, where the closed-form was derived. The authors investigated the distribution of uplink sum-rate asymptotically in terms of number of LIS elements. A similar channel modeling is  used in our work, and additionally the failure of elements and links will be integrated into the closed-form of outage probability.
From a channel modelling perspective, the fading channel is time varying, which make it important to adopt the scenario of an outdated CSI, as the instant CSI is hard to obtain in real communication systems. Recent work  \cite{mensi2022performance} adopted imperfect with outdated channel estimates to  analyse the outage probability of the best RIS selection in  a RIS-assisted network for vehicular communications, while dealing with the problem of RIS selection in a single hop system.  
The closed form of outage probability was derived in several specific communication scenario including relays \cite{vicario2009opportunistic, michalopoulos2012amplify, kwon2017relay}, and recently RISs \cite{mensi2022performance, guo2020outage}. Similar CSI model of \cite{mensi2022performance} was adopted in our work, where we assumed equal amplitude between elements of same RIS  to obtain a closed form expression of SNR and thus of the outage. The consideration of faulty elements for path selection was first  introduced in our previous work \cite{bensalem2023towards}, and to the best of our knowledge, this work is the first to consider the failure of RIS elements and connection degradation to analyse the outage probability of the RIS-aided path.  Without loss of generality the obtained outage probability closed form can be applied to both single hop and multi-hop system.



%
%
%
%
%


\section{System Model}
\addtolength{\topmargin}{-0.02in}


%

 We assume a communication scenario, consisting of a single antenna transmitter (TX) representing a base station (B), N different RISs, and  a single antenna receiver (RX), representing a mobile user (U), as illustrated in Fig. \ref{fig:arch}.a. We assume that an element of a RIS can fail with a probability $p$ and a connection can degrade due the presence of an obstacle, as shown in Fig.\ref{fig:arch}.b.  It is assumed that each RIS is partitioned into $J$ sub-surface block, where each block is operating as an independed RIS located in the same position. This assumption allows the RIS to be used to serve mutiple users by configuring the phase shift of each RIS block in order to reflect the signal to a certain location, denoting by $R_{k,j}$ the j$^\text{th}$ block of the k$^\text{th}$ RIS. A RIS is assumed to have $M$ sub-wavelength-sized reflecting elements, where each element has a square shape of size $L\times L$. As the RIS has $J$ equal blocks, we denote by  $M'=\frac{M}{J}$ the number of elements in each block. When $J$ is set to $1$, the RIS is operating as a signle block serving a single user at a time.  The direct link between $B$ and $U$  is considered not available due to obstacles, and the line-of-sight LoS links consist of $U$-to-$R_{k,j}$ and $R_{k,j}$-to-$B$. Let the channel coefficients of the LoS links  from user $U$ by the base station $B$  of RIS block $R_{k,j}$, be  denoted as a vector $\textbf{h}_{R_{k, j}B}=[h_{R_{k, j}B}^{1},..h_{R_{k, j}B}^{e}, .. h_{R_{k, j}B}^{M'}] $, where $h_{R_{k, j}B}^{e}$ is the channel coefficient related to the element $e$, and $\textbf{h}_{UR_{k, j}}=[h_{UR_{k, j}}^{1},...,h_{UR_{k, j}}^{e},...,h_{UR_{k, j}}^{M'}] $, where $k\in \lbrace 1,..., N \rbrace$,  $j \in  \lbrace 1,..., J \rbrace$, and  $e \in  \lbrace 1,..., M' \rbrace$.

The received signal at the user $U$ from base station $B$ through the RIS block $R_{k,j}$ is expressed as:
\begin{multline}
y_{U(R_{k,j})} =  \textbf{h}_{UR_{k, j}}^{T} \textbf{C}_{(k,j)}^{\frac{1}{2}} \Psi_{(k,j)} \textbf{C}_{(k,j)}^{\frac{1}{2}}  \textbf{h}_{R_{k, j}B} x + N_{U(R_{k,j})}\\
= \left( \sum_{l=1}^{M'}\sum_{s=1}^{M'}\sum_{m=1}^{M'}\sqrt{a_{l,s(k,j)}a_{l,m(k,j)}} h_{UR_{k, j}}^{s} h_{R_{k, j}B}^{m} e^{i\Psi_{m(k,j)}}  \right )\\ \times x  + N_{U(R_{k,j})}
\end{multline}
where $x$ is the transmitted signal  received by the user U and sent from the base station B with transmit power  P, $N_{U(R_{k,j})}$ is the AWGN power at user U.  We denote by  $ \textbf{C}_{(k,j)} = \{a_{l,s(k,j)}\}_{l,s\in [1,M']^2}$ the correlation coefficients matrix, and  $a_{l,s(k,j)}$ is the correlation coefficient  between the $l^{th}$ and $s^{th}$ elements of $j^{th}$ block of the $k^{th}$ RIS  $R_{k,j}$, such that $a_{l,s(k,j)} = 1$ if $l=s$, and $a_{l,s(k,j)}\in [0,1], \forall i,j \in [\![1,n]\!]^2$. $\Psi_{(k,j)}$ is the phase matrix induced by $R_{k,j}$, and  $\Psi_{m(k,j)}$ describes the $m^{th}$ adjustable phase induced by the  $j^{th}$ block of the $k^{th}$ RIS  $R_{k,j}$. 

  \begin{figure}[H]
 \centering 
   \includegraphics[scale=0.29]{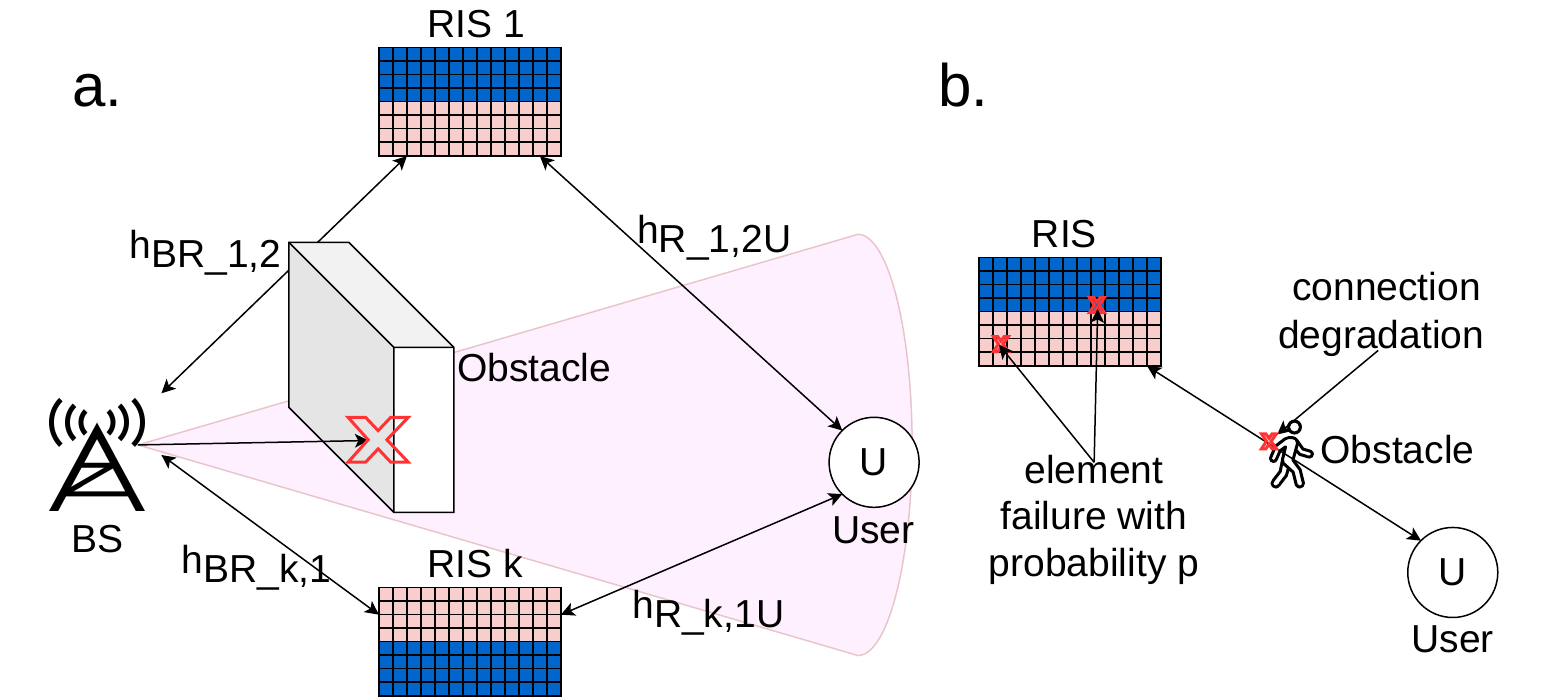}
 \caption{The reference  RIS-assited wireless network system: a. system architecture, b. hardware failure and connection degradation case.}
\label{fig:arch}
\end{figure}
\subsection{CSI Model}

The channel coefficients of the $U$-to-$R_{k,j}$ and $R_{k,j}$-to-$B$ links through the element $e$ are defined as follows:
 \begin{equation}\label{eq:channelURIS}
 \begin{split}
 h_{UR_{k, j}}^{e} = &g_{UR_{k, j}} e^{-i\theta_{k,j}^{e}}\sqrt{(d_{UR_{k, j}})^{-\delta_{k,j}(f)}},\\ 
 h_{R_{k, j}B}^{e}= & g_{R_{k, j}B} e^{-i\Phi_{k,j}^{e}}\sqrt{(d_{R_{k, j}B})^{-\delta_{k,j}(f)}}
 \end{split}
 \end{equation}
where $g_{UR_{k, j}}$ and $g_{R_{k, j}B}$ represents, respectively, the  gains for the $U$-to-$R_{k,j}$ and $R_{k,j}$-to-$B$ links, assuming that all elements of the block has the same gain,  $e^{-i\theta_{k,j}^{e}}$ and $ e^{-i\Phi_{k,j}^{e}}$  represents the normalized  LoS channel states,  with distances $d_{UR_{k, j}}$ and $d_{R_{k, j}B}$  and the  channel phases $\theta_{k,j}^{e}$ and $\Phi_{k,j}^{e}$ for the $U$-to-$R_{k,j}$ and $R_{k,j}$-to-$B$ links through the element $e$, $\delta_{k,j}(f)$ is the path loss exponent for the operating frequency $f$ of the signal.\\
We assume that the distance between RIS blocks is much lower that the distance between the base station $B$ and RISs and between the user $U$ and  RISs, thus the differene of power gain between the blocks of the same RIS is negligeable and can be considered equal.

The channel coefficient at the user $U$ through RIS block $R_{k,j} $ can be assumed as, similar to \cite{mensi2022performance}, but assuming that the gain is similar at all elements:

\begin{multline}\label{eq:my_channel}
h_{U(R_{k, j})} =   g_{UR_{k, j}} g_{R_{k, j}B} \\ \times  \left( \sum_{l=1}^{M'}\sum_{s=1}^{M'}\sum_{m=1}^{M'}   \sqrt{a_{l,s(k,j)}a_{l,m(k,j)}(d_{UR_{k, j}}d_{R_{k, j}B})^{-\delta_{k,j}}} \right. \\ \times \left. e^{i(\Psi_{m(k,j)}-\theta_{k,j}^{s}-\Phi_{k,j}^{m})}  \right ) 
\end{multline}
%

\subsection{SNR model}
The instantaneous SNR between the  user $U$ and the block $j$ of RIS $k$ is given by:
\begin{equation}
\iota_{UR_{k, j}}= \frac{P \nu |h_{UR_{k, j}}|^2}{N_0 }
\end{equation}
where the transmitter sends with source power $P$ and $N_0$ denotes the additive white Gaussian noise (AWGN) power, $\nu$ is a coefficient related to the impact of the existence of an obstacle, and $h_{UR_{k, j}}$ is defined by eq. (\ref{eq:channelURIS}).

Considering the selected path from the base station B to the user U, and using the channel definition in eq. (\ref{eq:my_channel}),  the instantaneous SNR corresponding to  the block $j$ of RIS $k$ is given as:
\begin{multline}\label{eq:finalsnr}
\overline{\iota}_{U(R_{k, j})}=  \frac{P \nu | h_{U(R_{k, j})}|^2}{N_0} \\ = \Omega\cdot  |g_{UR_{k, j}}|^2 |g_{R_{k, j}B}|^2  |\sum_{l=1}^{M'}\sum_{s=1}^{M'}\sum_{m=1}^{M'}\zeta_{l,s,m, U}^{k,j}   |^2
\end{multline}
where $\Omega$ and $\zeta_{l,s,m, U}^{k,j}$ are defined as follows:
\begin{multline}
 \Omega= \frac{P\nu}{N_0} \;\;\;\; \text{and} \\ \zeta_{l,s,m, U}^{k,j} =  \sqrt{a_{l,s(k,j)}a_{l,m(k,j)}(d_{UR_{k, j}}d_{R_{k, j}B})^{-\delta_{k,j}}} \\ \times e^{i(\Psi_{m(k)}-\theta_{k,j}^{s}-\Phi_{k,j}^{m})}  
 \end{multline}
After simplifying the term with square of summations,  the previous form can be re-written as follows:
\begin{equation}\label{eq:finalsnr2}\begin{split}
\overline{\iota}_{U(R_{k, j})}= & \Omega \cdot    |g_{UR_{k, j}} |^2 |g_{R_{k, j}B}|^2  \left [ \sum_{l=1}^{M'}\sum_{s=1}^{M'}\sum_{m=1}^{M'}|\zeta_{l,s,m, U}^{k,j}|^2 \right. \\ & \left. +  2 \sum_{ (l,s,m,l',s',m')\in \Xi}^{M'} \eta_{l,s,m,l',s',m', U}^{k,j}  \right ]
\end{split}
\end{equation}
where $
|\zeta_{l,s,m, U}^{k,j}|^2 = a_{l,s(k,j)}a_{l,m(k,j)}(d_{UR_{k, j}}d_{R_{k, j}B})^{-\delta_{k,j}}, $

\begin{multline}
\eta_{l,s,m,l',s',m', U}^{k,j} =  \zeta_{l,s,m, U}^{k,j} \zeta_{l',s',m', U}^{k,j}\\=  \sqrt{a_{l,s(k,j)}a_{l,m(k,j)}a_{l',s'(k,j)}a_{l',m'(k,j)}} (d_{UR_{k, j}}d_{R_{k, j}B})^{-\delta_{k,j}}\\ \times cos((\Psi_{m(k,j)} -\theta_{k,j}^{s} -  \Phi_{k,j}^{m}) - ( \Psi_{m'(k,j)}  -\theta_{k,j}^{s'} - \Phi_{k,j}^{m'} )  )
\end{multline}
and  $\sum_{ (l,s,m,l',s',m')\in \Xi}^{M'} = \sum_{l=1}^{M'}\sum_{s=1}^{M'}\sum_{m=1}^{M'} \sum_{l'=1}^{M'}\sum_{s'=1}^{M'}\sum_{m'=1}^{M'}$ where  $\Xi$ follows the conditions given as:
\begin{equation}
\Xi  =  \begin{cases}
       (l \neq l') \cap (s \neq s') \cap (m \neq m')  \\
      m < l'\\
      ( m = l') \cap (l < s ) \cap (m'\leq s') \\
      ( m \leq l') \cap (l \leq s) \cap (m'> s')
    \end{cases}
\end{equation}
\subsection{Element failure model}
We consider the possibility of RIS element failure. Let $1_{k,j}^{e}$ be an indicator of whether an element $e$ fails or not; $1_{k,j}^{e}=1$ if element $e$ in RIS block $R_{k, j}$ does not fail and $0$ otherwise. We denoted by $P[1_{k,j}^{e}=0]=p, \forall e \in [1,...,M']$  the probability that element $e$ of a RIS block fails, $Q$ be the number failed elements. 
\begin{prop} The probability mass function  $P_{Q;M'}(q)$ that $q$ element out of $M'$ fails is given as follows:
\begin{equation}
P(Q=q)=C_{M'}^{q} p^{q}(1-p)^{M'-q}
\end{equation}
\end{prop}
\begin{proof}
The random variable $Q$ is assumed to follow the binomial distribution with parameters $M'\in \mathbb{N}$ and $p\in [0,1]$, $Q\sim \mathcal{B}(M',p)$. The probability of getting exactly $q$ failed element is equivalent to getting $q$ successes in $M'$ independent Bernoulli trials, with the same rate $p$, and can be given by the probability mass function. 
\end{proof}
\subsection{Path selection}
We assume that  the base station $B$ has a perfect knowledge of the RF channels, which allows it to select the RIS providing the highest CSI. The base station periodically receives CSIs feedback from the $N$  RISs  
considering that some RISs could  fail due to an obstacle.

In this paper, we consider that a  RIS  Controller can read the signal received at the RIS with an active element, decide the best RIS block to select and forward the information to the base station $B$ to adjust the arrays accordingly. We assume that only the user to RIS link is considered for the selection, and the selection RIS block is given by:
\begin{equation}
(k^{*},j^{*}) = \arg \max_{k,j} (\overline{\iota}_{UR_{k, j}}(T))
\end{equation}
 Thus, the selected path in a 1 hop scenario, would be $B\rightarrow R_{k^*, j^*}\rightarrow U $.
%

\section{Outage Probability Analysis}
The outage probability of the communication channel is defined as the probability that  a certain target data rate cannot be supported by the system. It can be computed as the  the probability that the overall SNR $\overline{\iota}_{U(R_{k, j})}$ being  lower than a threshold, $\gamma_T= 2^{2r}-1$, where $r$ represents    the target rate.

\subsection{Outage considering connection degradation }
\begin{prop} The outage probability of a RIS-assisted path is given by eq. (\ref{eq:outage2}),   where $K_1(\cdot)$ is the 1$^\text{st}$ order modified Bessel function of the second kind. 

\end{prop}

\begin{proof}
In our mobile scenario, the outage probability is defined as the cdf  $F_{\overline{\iota}_{U(R_{k, j})}}(\gamma_T)$, given by:
\begin{multline}\label{eq:outage1}
F_{\overline{\iota}_{U(R_{k, j})}}(\gamma_T)
 =  \text{ Pr }(\overline{\iota}_{U(R_{k, j})}< \gamma_T )\\=\text{ Pr }( |g_{UR_{k, j}} |^2 |g_{R_{k, j}B}|^2 < \frac{\gamma_T}{ \Upsilon_{U,k,j} }  )\\=\int_{0}^{\infty} \text{Pr }\left(  |g_{UR_{k, j}} |^2 y < \frac{\gamma_T}{ \Upsilon_{U,k,j} }  \right)  f_{|g_{R_{k, j}B}|^2}(y) dy\\ =\int_{0}^{\infty} F_{|g_{UR_{k, j}} |^2 } \left(  \frac{\gamma_T}{ y \Upsilon_{U,k,j} }  \right)  f_{|g_{R_{k, j}B}|^2}(y) dy
 \end{multline}
 where the parameter $\Upsilon_{U,k,j}$ is defined as:
 \begin{equation}\begin{split}
 \Upsilon_{U,k,j} =& \Omega  \left [ \sum_{l=1}^{M'}\sum_{s=1}^{M'}\sum_{m=1}^{M'}|\zeta_{l,s,m, U}^{k,j}|^2  \right. \\&   \left. +  2 \sum_{ (l,s,m,l',s',m')\in \Xi}^{M'} \eta_{l,s,m,l',s',m', U}^{k,j}  \right ]
 \end{split}
\end{equation}

We  assume that  the outdated CSI $\tilde{h}_{UR_{k, j}}  $ and $\tilde{h}_{R_{k, j}B}  $  are following the Rayleigh distribution \cite{bjornson2020rayleigh}, assuming that all elements of the same RIS block have the same CSI. Assuming that  $ |\tilde{g}_{UR_{k, j}} |^2 \sim \text{Exp(}\frac{1}{\lambda_{U}}\text{)}$ where   $\frac{1}{\lambda_{U}}$ represents the rate parameter i.e. $\lambda_{U}$ represents the average SNR per link. In addition, $ |\tilde{g}_{R_{k, j}B} |^2 \sim \text{Exp(}\frac{1}{\lambda_{B}}\text{)}$ where   $\frac{1}{\lambda_{B}}$ is  the rate parameter. 
 The PDF and the CDF of the outdated CSI corresponding to one  reflecting element from the block $j$ of RIS $k$  of the first link from user U to RIS are given by:
 \begin{equation}\label{eq:pdf1}
 f_{|\tilde{g}_{UR_{k, j}} |^2}(x) = \frac{1}{\lambda_{U}} e^{-\frac{x}{\lambda_{U}}}
 \end{equation}
 \begin{equation}\label{eq:cdf1}
 F_{|\tilde{g}_{UR_{k, j}} |^2}(x) = 1 -  e^{-\frac{x}{\lambda_{U}}}
 \end{equation}
The distributions of the outdate CSI  $|\tilde{g}_{UR_{k, j}} |^2 $ and the instantaneous CSI  $|g_{UR_{k, j}} |^2 $ are correlated exponential distributions \cite{vicario2009opportunistic}, thus the joint pdf is given by:
\begin{equation}\label{eq:jointpdf}
 f_{|g_{UR_{k, j}} |^2, |\tilde{g}_{UR_{k, j}} |^2}(x,y) =\frac{ e^{-\frac{x+y}{(1-\rho_1^{2}) \lambda_{U}}}}{(1-\rho_1^{2})\lambda_{U}^{2}} I_0\left (  \frac{2\sqrt{\lambda_{U}^{2} xy}}{(1-\rho_1^{2}) \lambda_{U}} \right )
 \end{equation}
 where $I_0(·)$ is  the zero-order modified Bessel function of the first kind.
 
 \begin{strip}
 \begin{equation}\label{eq:outage2}
\begin{aligned}
F_{\overline{\iota}_{U(R_{k, j})}}(\gamma_T)=&
  \frac{N^2J^2  (1-\rho_1^{2})  }{\lambda_{B} }   \sum_{s=0}^{NJ -1}  \sum_{s'=0}^{NJ -1}    \frac{(-1)^{s+s'} C^{s}_{NJ-1} C^{s'}_{NJ-1}}{(\lambda_{U}^2 + s+s\rho_1^{2}+1)(s'+s'\rho_2^{2}+1)}   \left[  \frac{(s'+s'\rho_2^{2}+1)(1-\rho_2^{2}) \lambda_{B}}{(\lambda_{B}^2 + s'+s'\rho_2^{2}+1)} \right. \\ &- \sqrt{\frac{4(\lambda_{U}^2 + s+s\rho_1^{2}+1) (s'+s'\rho_2^{2}+1)(1-\rho_2^{2}) \lambda_{B}\gamma_T}{ (\lambda_{B}^2 + s'+s'\rho_2^{2}+1)(s+s\rho_1^{2}+1) (1-\rho_1^{2})\lambda_{U}   \Upsilon_{U,k,j}}} \\&\left.  \times  K_{1}\left( \sqrt{\frac{4(\lambda_{U}^2 + s+s\rho_1^{2}+1)(\lambda_{B}^2 + s'+s'\rho_2^{2}+1)\gamma_T }{ (s+s\rho_1^{2}+1)(1-\rho_1^{2})(s'+s'\rho_2^{2}+1)(1-\rho_2^{2}) \lambda_{B} \lambda_{U}   \Upsilon_{U,k,j} }   } \right)  \right] 
\end{aligned}
    \end{equation}
    \end{strip}

 The pdf of the instantaneous SNR of the link between the user U and the block $j$ of RIS $k$ is given as:
 \begin{equation}\label{eq:pdf11} 
 f_{|g_{UR_{k, j}} |^2}(x) = \int_{0}^{\infty}  f_{|g_{UR_{k, j}} |^2 |\;|\tilde{g}_{UR_{k, j}} |^2}(x|y) f_{|\tilde{g}_{UR_{k, j}} |^2 }(y) dy
 \end{equation}
 The conditional pdfs $ f_{|g_{UR_{k, j}} |^2 |\;|\tilde{g}_{UR_{k, j}} |^2}(.|.) $ are all identical for all $k,j$ due to the fact that $|g_{UR_{k, j}} |^2 $ and  $|\tilde{g}_{UR_{k, j}} |^2$  are i.i.d. RVs. Thus the conditional pdfs can be expressed as follows:
 \begin{equation}\label{eq:condpdf}
  f_{|g_{UR_{k, j}} |^2 |\;|\tilde{g}_{UR_{k, j}} |^2}(x|y)= \frac{ f_{|g_{UR_{k, j}} |^2, |\tilde{g}_{UR_{k, j}}|^2}(x,y)}{f_{|\tilde{g}_{UR_{k, j}}|^2}(y) }
 \end{equation}
 Thus, we substitute  (\ref{eq:pdf1})   and (\ref{eq:jointpdf})   into (\ref{eq:condpdf}) to obtain the following expression:
 \begin{equation}\label{eq:condpdf1}
 f_{|g_{UR_{k, j}} |^2 |\;|\tilde{g}_{UR_{k, j}} |^2 }(x|y)= \frac{ e^{-\frac{x+\rho_1^{2}y}{(1-\rho_1^{2}) \lambda_{U}}}}{(1-\rho_1^{2})\lambda_{U}} I_0\left (  \frac{2\sqrt{\lambda_{U}^{2} xy}}{(1-\rho_1^{2}) \lambda_{U}} \right )
 \end{equation}
 The cdf of outdated CSI $|\tilde{g}_{UR_{k, j}} |^2 $ can be given as follows:
 \begin{multline}
 F_{|\tilde{g}_{UR_{k, j}} |^2 }(x) = \sum_{i=1}^{N}\sum_{s=1}^{J}\text{Pr}(|\tilde{g}_{UR_{k, j}} |^2  \leq x \text{ } \cap  k=i, j=s)\\
 =  NJ \text{  Pr}(|\tilde{g}_{UR_{k, j}} |^2  \leq x \text{ } \cap  k=i, j=s)\\=  NJ\int_{0}^{x}f_{|\tilde{g}_{UR_{k, j}} |^2 }(y) \text{   Pr}(  k=i, j=s| \; |\tilde{g}_{UR_{k, j}} |^2  \leq y )dy 
 \end{multline}
 Given that the best selection considers only the link between the user $U$ and the RIS devices, the link between the base station $B$ and the RISs is not taken into consideration. In addition, we consider all paths have the same distribution and using eq. (\ref{eq:cdf1}), the cdf can be given as:
 \begin{multline}\label{eq:cdf2}
 F_{|\tilde{g}_{UR_{k, j}} |^2 }(x) = NJ\int_{0}^{x}f_{|\tilde{g}_{UR_{k, j}} |^2 }(y) [F_{\tilde{h}_{UR}}(y)]^{NJ-1}dy\\=  NJ\int_{0}^{x}f_{|\tilde{g}_{UR_{k, j}} |^2 }(y) [1 -  e^{-\frac{y}{\lambda_{U}}}]^{NJ-1}dy\\
 =  NJ \sum_{s=0}^{NJ-1}(-1)^{s} C^{s}_{NJ-1}\int_{0}^{x}f_{|\tilde{g}_{UR_{k, j}} |^2 }(y) e^{-\frac{sy}{\lambda_{U}}}dy\\
 =  NJ \sum_{s=0}^{NJ-1} \frac{(-1)^{s} C^{s}_{NJ-1}}{s+1}(1-e^{-\frac{(s+1)x}{\lambda_{U}}})
 \end{multline}
 Hence, we differentiate the obtained cdf $F_{|\tilde{g}_{UR_{k, j}} |^2 }(x) $ to get the outdated CSI pdf given by:
 \begin{equation}\label{eq:pdf2}
  \begin{split}
f_{|\tilde{g}_{UR_{k, j}} |^2 }(x)=& \frac{d }{dx} F_{|\tilde{g}_{UR_{k, j}} |^2 }(x) \\=&   \frac{NJ}{\lambda_{U}}  \sum_{s=0}^{NJ-1} (-1)^{s} C^{s}_{NJ-1} e^{-\frac{(s+1)x}{\lambda_{U}}}
 \end{split}
 \end{equation}
 By replacing eq. (\ref{eq:pdf2}) and eq. (\ref{eq:condpdf1}) into eq. (\ref{eq:pdf11})  we derive the pdf of the instantaneous CSI as:
 
 \begin{multline}\label{eq:pdf11}
 f_{|g_{UR_{k, j}} |^2}(x) =  \int_{0}^{\infty}  f_{|g_{UR_{k, j}} |^2 |\tilde{g}_{UR_{k, j}} |^2 }(x|y) f_{|\tilde{g}_{UR_{k, j}} |^2 }(y) dy\\
 =  \frac{NJ}{(1-\rho_1^{2})\lambda_{U}^2 }  \sum_{s=0}^{NJ-1} (-1)^{s} C^{s}_{NJ-1}\int_{0}^{\infty}   e^{-\frac{x+(s+s\rho_1^{2}+1)y}{(1-\rho_1^{2}) \lambda_{U}}} \\ \times I_0\left (  \frac{2\sqrt{\lambda_{U}^{2} xy}}{(1-\rho_1^{2}) \lambda_{U}} \right ) dy
 \end{multline}
 
 Using the identity in  \cite{gradshteyntables} eq. (6.643.4), where we set $n=0, \nu=0, \alpha=\frac{s+s\rho_1^{2}+1}{(1-\rho_1^{2}) \lambda_{U}}, \beta \frac{\sqrt{x}}{1-\rho_1^{2}}$,  the closed form of the pdf is given as:
  \begin{equation}\label{eq:pdf3}
 f_{|g_{UR_{k, j}} |^2}(x) 
 = \frac{NJ }{\lambda_{U} }  \sum_{s=0}^{NJ-1}   \frac{(-1)^{s} C^{s}_{NJ-1}}{s+s\rho_1^{2}+1}  \times e^{-\frac{(\lambda_{U}^2 + s+s\rho_1^{2}+1)x}{(s+s\rho_1^{2}+1)(1-\rho_1^{2}) \lambda_{U}}} 
 \end{equation}
 Using the closed form of the pdf in (\ref{eq:pdf3}),  The cdf of $|g_{UR_{k, j}} |^2$ is given as:
 \begin{multline}\label{eq:cdf3}
 F_{|g_{UR_{k, j}} |^2}(x) 
 = NJ (1-\rho_1^{2})   \sum_{s=0}^{NJ -1}   \frac{(-1)^{s} C^{s}_{NJ-1}}{\lambda_{U}^2 + s+s\rho_1^{2}+1} \\ \times \left( 1-e^{-\frac{(\lambda_{U}^2 + s+s\rho_1^{2}+1)x}{(s+s\rho_1^{2}+1)(1-\rho_1^{2}) \lambda_{U}}} \right) 
 \end{multline}
 Due to symmetry, the closed form of the CSI  pdf and cdf of  $|g_{UR_{k, j}} |^2$ and $|\tilde{g}_{UR_{k, j}} |^2$ is given by (\ref{eq:pdf3}) and (\ref{eq:cdf3}), respectively, by interchaning $\rho_1$ with $\rho_2$ and $\lambda_U$ with $\lambda_B$.\\
 Hence, we can evaluate the outage probability by  substituting (\ref{eq:cdf3}) and (\ref{eq:pdf3}) into (\ref{eq:outage1}) and using the identity in  \cite{gradshteyntables} eq. (3.324.1), where we set $\gamma=\frac{(\lambda_{B}^2 + s'+s'\rho_2^{2}+1)}{(s'+s'\rho_2^{2}+1)(1-\rho_2^{2}) \lambda_{B}}$ and $ \beta = \frac{4(\lambda_{U}^2 + s+s\rho_1^{2}+1)\gamma_T }{ (s+s\rho_1^{2}+1)(1-\rho_1^{2}) \lambda_{U}   \Upsilon_{U,k,j} }$,  yielding eq. (\ref{eq:outage2}).
 

\end{proof}
\subsection{Outage considering Element failure}
Considering that in practical system, the RIS hardware is vulnerable to hardware failures, we assume that the outage probability can be dependent on the state of the RIS elements. 

 Without loss of generality, we assume that the first $q$ elements fail, which means that we have $M'-q$ operating elements in RIS block $R_{k, j}$ starting from element index $q+1$.\\ 
 Using the theorem of total probability, the outage probability $F_{\overline{\iota}_{U(R_{k, j})}}(\gamma_T)$ can be given by:

\begin{equation}
F_{\overline{\iota}_{U(R_{k, j})}}(\gamma_T)=\sum_{q=0}^{M'} P(Q=q) F_{\overline{\iota}_{U(R_{k, j})}}^{q}(\gamma_T)
\end{equation}
where $F_{\overline{\iota}_{U(R_{k, j})}}^{q}(\gamma_T)$ is the outage probability when $k$ elements of the RIS block $R_{k, j}$ are failed. 
%
\section{Numerical Evaluation}\label{sec:results}
In this section, numerical results are presented to verify and validate the correctness of the obtained closed form expressions and evaluate the impact of the main system parameters on the  performance and reliability. We consider a set of $N=4$ RISs, with $M=30$ elements, and $J=4$ partitions. We consider a source power $P=30$ dB, a noise AWGN power $N_0 = 10$ dB, a correlation coefficient between instantaneous and outdated CSI of user-RIS link and RIS-BS link $\rho_1=\rho_2=0.1$, the distance $d_{UR_{k, j}}=d_{R_{k, j}B}=4$ m, and  the phase shifts are random.

Fig. \ref{fig:snr_out_p} shows the outage of the best RIS as a function of the normalized average SNR of U-RIS and RIS-B, for different value of element failure probability. The results shows  that the outage probability decreases as the SNR increases for all values of failure. Also, it increase with the increase of  failure probability $p$  of each RIS element, which quantifies how a lower number of elements is impacting the transmitted signal.


 \begin{figure}[ht]
\centering
\includegraphics[scale=0.34]{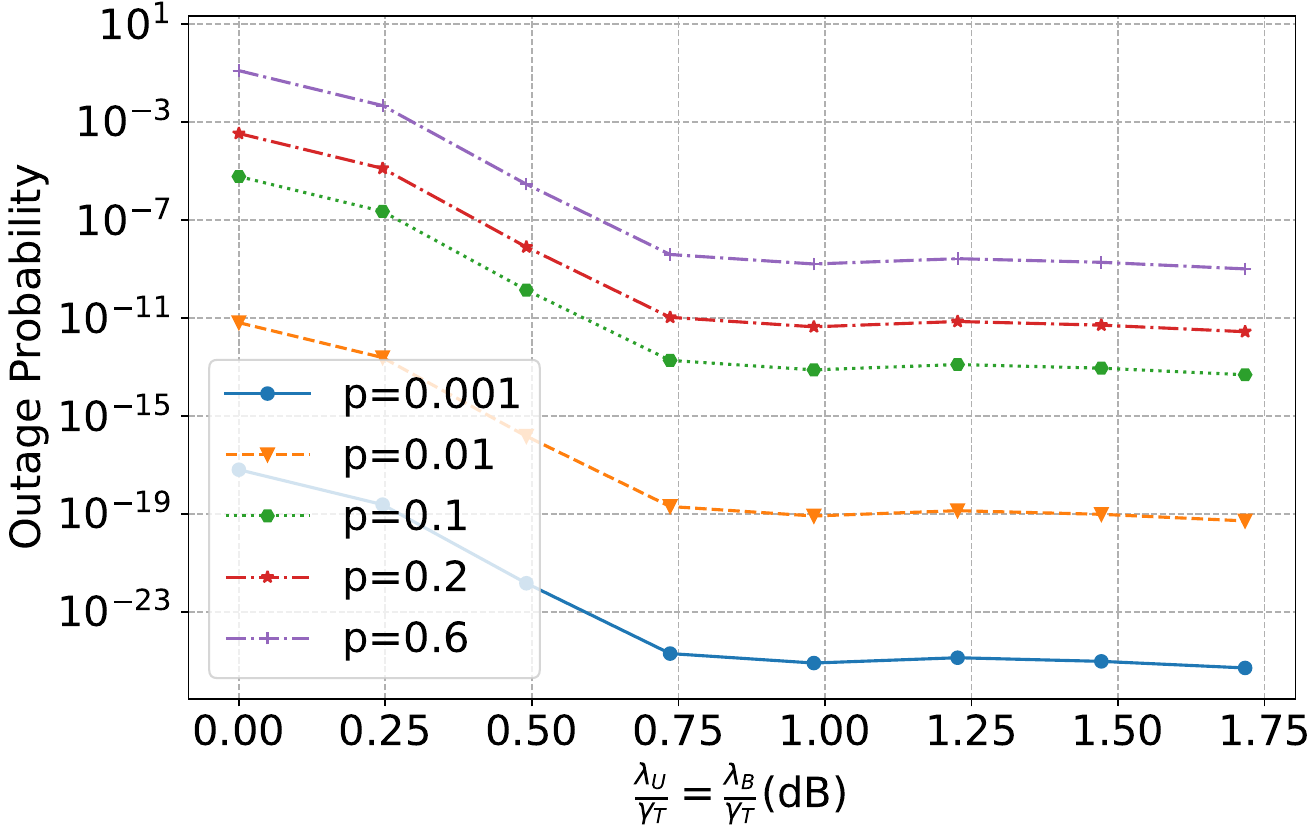}  
 \caption{Outage probability of best RIS vs the normalized average SNR of U-RIS and RIS-B, with different probability of failures, for $d_{UR_{k,j}} = 4$ m.}
\label{fig:snr_out_p}
\vspace{-0.2 cm}
\end{figure}
Fig. \ref{fig:snr_out_pw} shows the outage of the best RIS as a function of the normalized average SNR of U-RIS and RIS-B, for different value of element failure probability and considering the presence or abscence of an obstacle with $\nu=0.1$. It is observed that the presence of obstacle start impacting the outage when the average SNR is slighly higher than $0.5$ dB, and a higher element failure cause a higher outage probability.
  \begin{figure}[ht]
\centering
\includegraphics[scale=0.34]{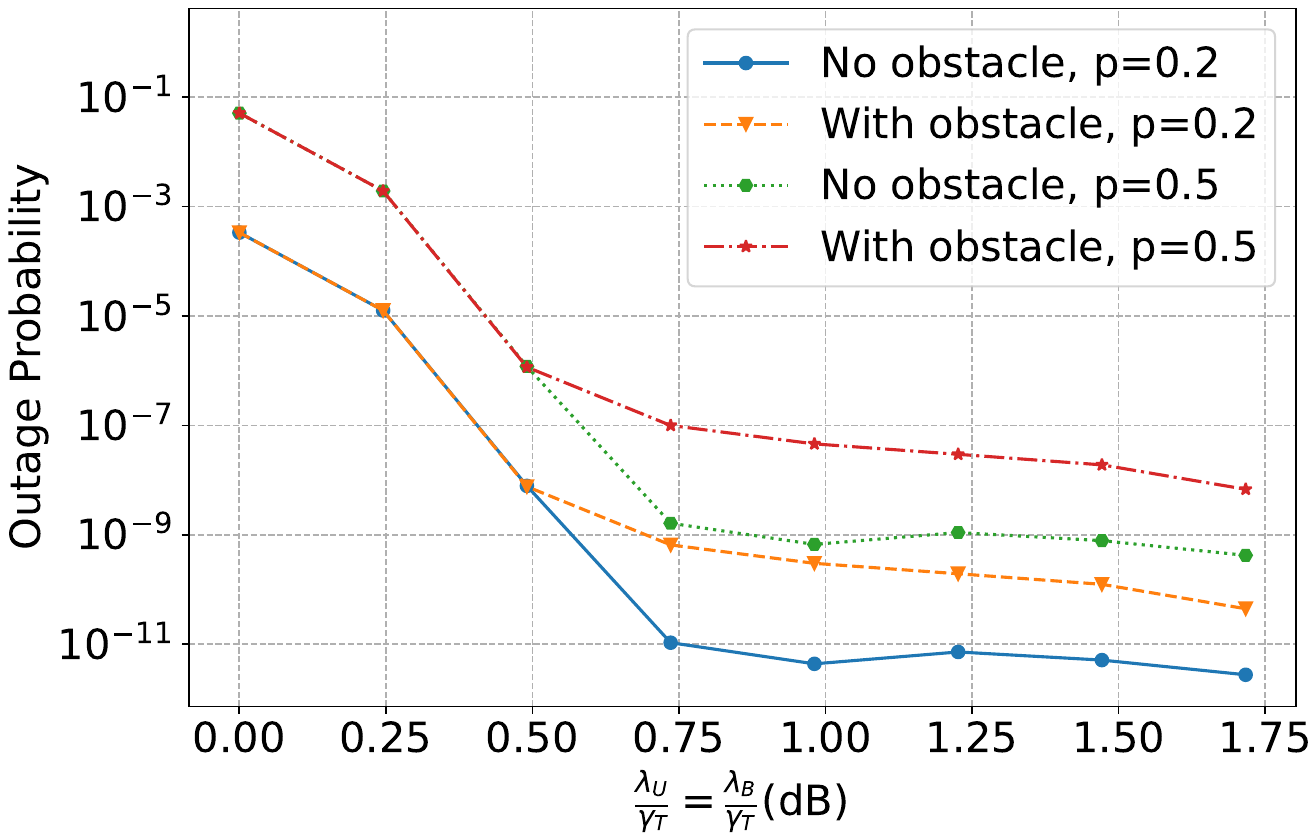}   
 \caption{Outage probability of best RIS vs the normalized average SNR of U-RIS and RIS-B, with different reliability scenarios, for $d_{UR_{k,j}} = 4$ m.}
\label{fig:snr_out_pw}
\vspace{-0.2 cm}
\end{figure}

 Fig. \ref{fig:M_out} shows the outage of the best RIS as a function of number of elements, where $\gamma_T=3$  and $\lambda_U=0.05$, considering different values of $\rho$,  the correlation coefficient between outdated and instantaneous CSI, and correlation between elements. For weakly correlated elements and  $\rho\leq 0.5$, the outage is high for all sizes of RIS, and low when the correlation is strong. For strongly correlation elements, we observe that the outage  decreases at a certain number of elements, $M=50$ for a  4 m distance, meaning that the received power  achieves the target SNR at a certain RIS size. 
\begin{figure}[ht]
\centering
   \includegraphics[scale=0.34]{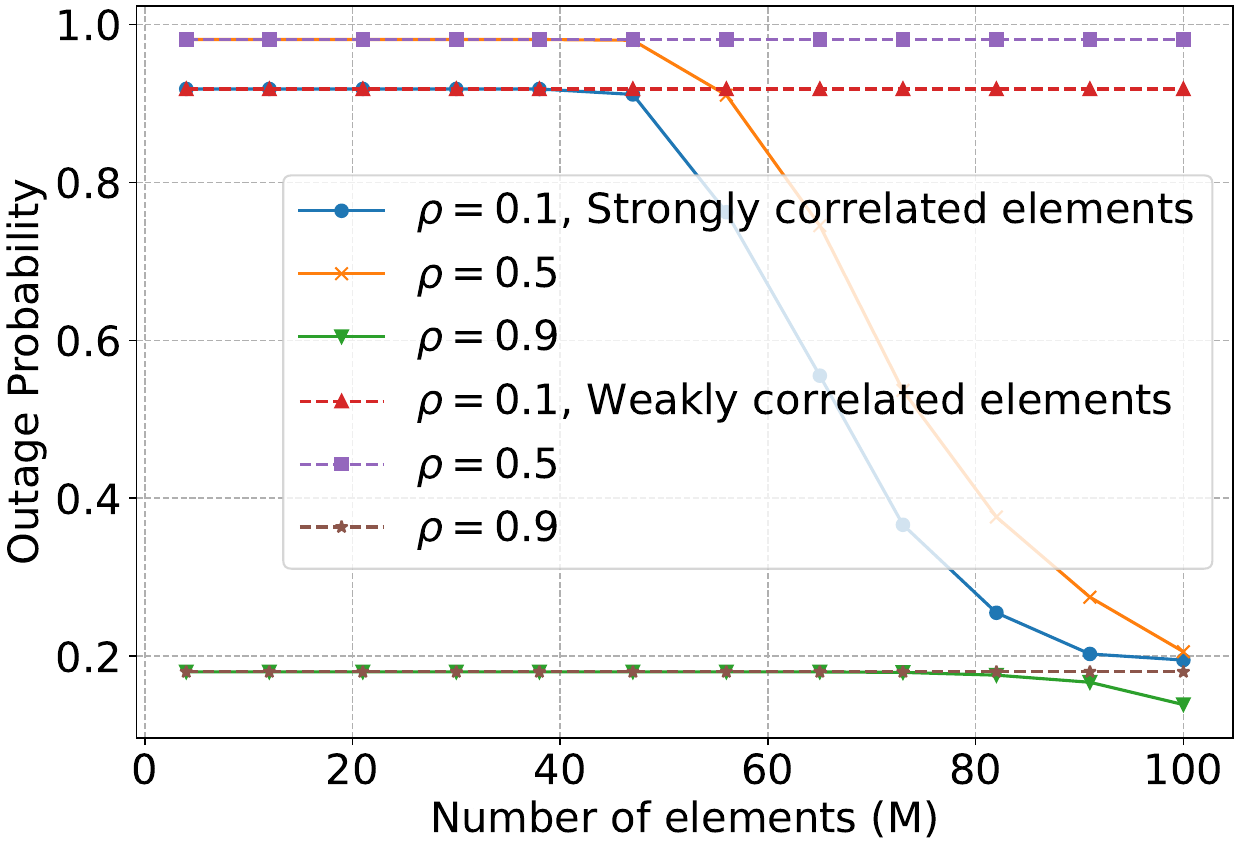}
 \caption{Outage probability of best RIS vs number of elements $M$}
\label{fig:M_out}
\vspace{-0.2 cm}
\end{figure}

Fig. \ref{fig:d_out} shows the outage of the best RIS as a function of user-RIS distance,  with and without obstacle. For  $\rho\leq 0.5$, the outage is high for distance higher than 3 m in presence of obstacle and 5 m without an obstacle. When $\rho=0.9$,  the correlation is strong and the outage is low for all distances.
\begin{figure}[ht]
\centering
   \includegraphics[scale=0.34]{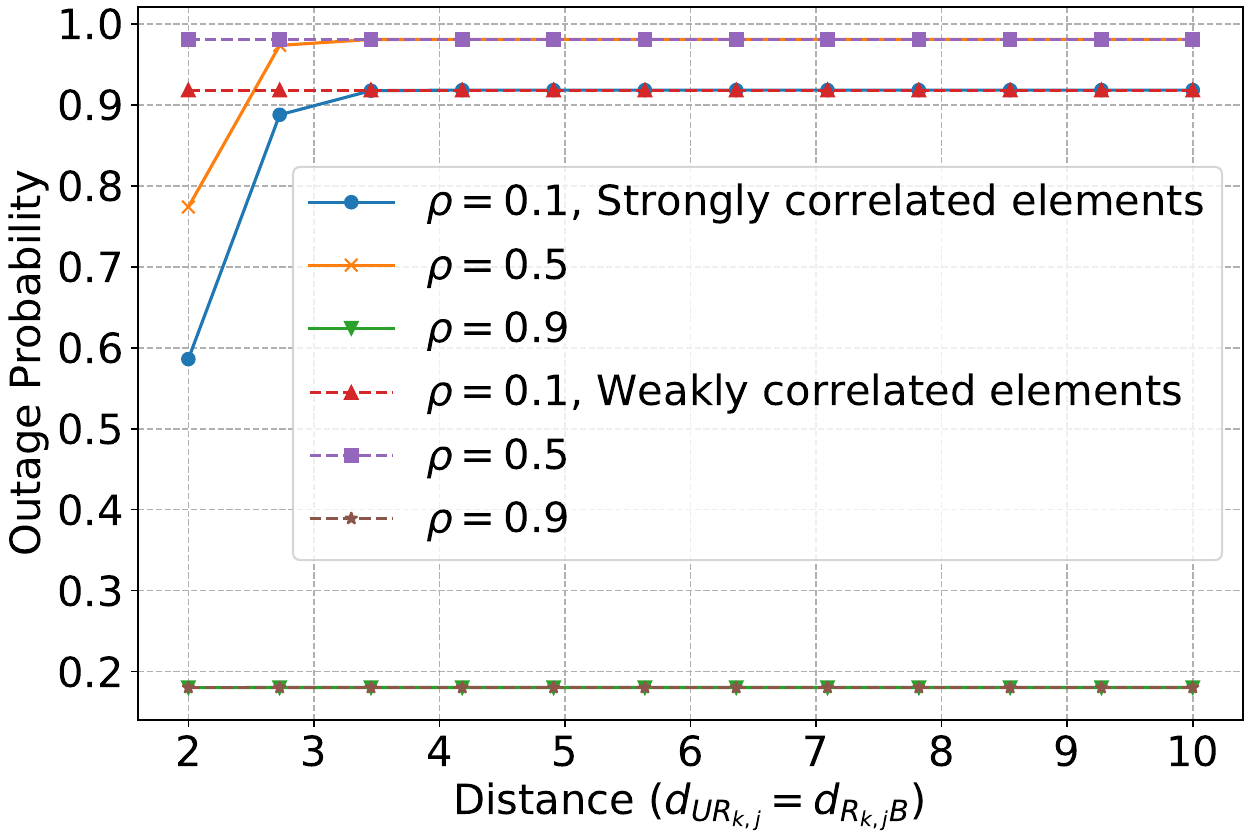}
  \includegraphics[scale=0.34]{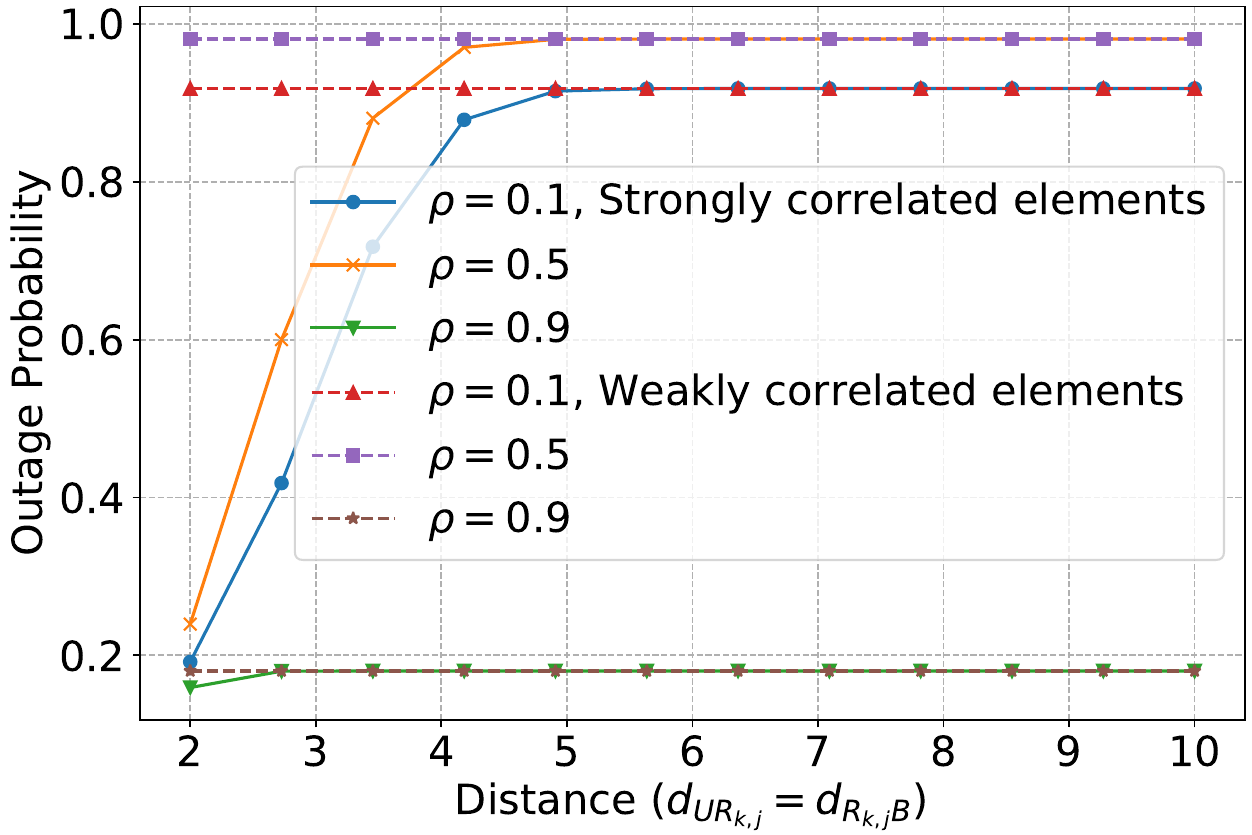}
 \caption{Outage probability of best RIS vs distance $d_{UR_{k,j}}$ with and without  obstacle, $\gamma_T=3$ and $\lambda_U=0.05$}
\label{fig:d_out}
\vspace{-0.2 cm}
\end{figure}

 \section{Conclusion}\label{sec:conclusion}
In this paper, we proposed  a novel analysis of the outage probability of wireless paths, considering  RIS hardware failure and  connection degradation due to obstacles. We analysed and obtained a closed-form of the instantaneous CSI model of a  communication system with a base station, RISs  and user.  We formulate the SNR model using the instantaneous CSI and analyse the partial best path selection, using the "best" RIS  selected based on outdated CSI. 
We derived a closed-form expression of the outage probability for an arbitrary path,  considering the effect of RIS hardware failure and signal degradation due to obstacles. 
The analysis can be used to quantify the impact of the  blockages and  the RIS  failure on the outage probability. The results shows at a SNR approaching zero,  the presence of obstacles and the failure probability of RIS elements highly affect the outage probability. 
a large enough RIS device vs. a communication distance provides a more reliable connection, considering  a RIS with strongly correlated elements and their likelihood of failure. 


\section*{Acknowledgment}
This work was partially supported by the DFG Project Nr. JU2757/12-1,  and by the Federal Ministry of Education and Research of Germany, joint project 6GRIC, 16KISK031.
\vspace{-0.2cm}
\bibliographystyle{IEEEtran}
\bibliography{mybib}

%

\end{document}